%% file: inv_subspaces.tex
\documentclass[
final,
authoryear
]{elsarticle}

\usepackage[elsarticle]{artmacs}
\usepackage[utf8]{inputenc}

\nopreprintlinetrue

\usepackage{tabularx}
\defcitealias{sage611}{\texttt{SageMath}}
\usepackage{subfig}
\usepackage{rotating}

\usepackage{tikz}
\usetikzlibrary{decorations.pathreplacing, decorations.text}
\usetikzlibrary{arrows}
\usetikzlibrary{external}

\newcommand{\Mmax}{\widehat{M}}

\newcommand{\expn}{expn}
\newcommand{\diag}{diag}
\newcommand{\gcrc}{gcrc}

\newcommand{\minpoly}{minpoly}
\newcommand{\supp}{supp}
\newcommand{\rk}{rk}
\newcommand{\nul}{nul}

\newcommand{\numchains}[1]{\operatorname{\#chains}(#1)}
\newcommand{\numdepth}[1]{\operatorname{\#depth}(#1)}

\newcommand{\depth}[1]{\operatorname{depth}(#1)}
\newcommand{\fmclc}{{f^{*}}}

\newcommand{\nxn}{{n\times n}}

\newcommand{\F}{{\mathsf{F}}}
\newcommand{\Fs}{\FF_{s}}
\newcommand{\Ft}{\FF_{t}}

\newcommand{\Fbar}{{\overline{\F}}}

\newcommand{\pP}[3][]{#2[x; #3]_{#1}}    

\newcommand{\cAq}{{\Fr[x;q]}}

\newcommand{\vV}[3][]{L[\sigma_{#2}; #3]_{#1}}    

\newcommand{\gf}{\mathsf{g}}
\newcommand{\remove}[1]{}

\begin{document}

\begin{frontmatter}

\title{Counting invariant subspaces\\and decompositions of additive polynomials}
\author[1]{Joachim von~zur~Gathen}
\ead{gathen@bit.uni-bonn.de}
\ead[url]{http://cosec.bit.uni-bonn.de/}
\author[2]{Mark Giesbrecht\corref{cor}}
\ead{mwg@uwaterloo.ca}
\ead[url]{https://cs.uwaterloo.ca/~mwg}
\author[3]{Konstantin Ziegler}
\ead{mail@zieglerk.net}
\ead[url]{http://zieglerk.net}
\address[1]{B-IT, Universit\"at Bonn\\
  D-53115 Bonn, Germany}
\address[2]{Cheriton School of Computer Science\\
  University of Waterloo, Waterloo, ON, N2L 3G1 Canada}
\address[3]{University of Applied Sciences Landshut\\
D-84036 Landshut, Germany}
\cortext[cor]{Corresponding author}

\begin{abstract}
  The functional (de)composition of polynomials is a topic in pure and
  computer algebra with many applications.  The structure of
  decompositions of (suitably normalized) polynomials $f = g \circ h$
  in $F[x]$ over a field $F$ is well understood in many cases, but
  less well when the degree of $f$ is divisible by the positive characteristic
  $p$ of $F$. This work investigates the decompositions of
  $r$-\emph{additive} polynomials, where every exponent and also the
  field size is a power of $r$, which itself is a power of $p$.

  The decompositions of an $r$-additive polynomial $f$ are intimately
  linked to the Frobenius-invariant subspaces of its root space $V$ in
  the algebraic closure of $\F$. We present an efficient algorithm to
  compute the rational Jordan form of the Frobenius automorphism on $V$. A
  formula of \cite{fri11} then counts the number of
  Frobenius-invariant subspaces of a given dimension and we derive the
  number of decompositions with prescribed degrees.
\end{abstract}


\end{frontmatter}

\section{Introduction}

The \emph{composition} of two polynomials $g,h \in F[x]$ over a field
$F$ is denoted as $f= g \circ h= g(h)$, and then $(g,h)$ is a
\emph{decomposition} of $f$. If $g$ and
$h$ have degree at least $2$, then $f$ is \emph{decomposable} and $g$
and $h$ are \emph{left} and \emph{right components} of $f$, respectively.

Since the foundational work of Ritt, Fatou, and Julia in the 1920s on
compositions over $\CC$, a substantial body of work has been concerned
with structural properties (e.g., \cite{frimac69}, \cite{dorwha74},
\cite{sch82c, sch00c}, \cite{zan93}), with algorithmic questions
(e.g., \cite{barzip85}, \cite{kozlan89b}), and more recently with
enumeration, exact and approximate (e.g., \cite{gie88b},
\cite{blagat13}, \cite{gat08c}, \cite{Ziegler2015,Ziegler2016}). A fundamental dichotomy
is between the \emph{tame case}, where the characteristic $p$ of $F$
does not divide $\deg g$, see \cite{gat90c}, and the \emph{wild case},
where $p$ divides $\deg g$, see \cite{gat90d}.

\cite{zip91} suggests that the block decompositions of \cite{lanmil85}
for determining subfields of algebraic number fields can be applied to
decomposing rational functions even in the wild case. \cite{bla13}
proves this formally and shows that this idea can be used to compute
all decompositions of a polynomial with an indecomposable right
component. \cite{gie98} provides fast algorithms for the decomposition
of \emph{additive} (or linearized) polynomials, where all exponents are
powers of $p$.
Subsequent improvements in the cost of
factorization and basic operations have been made in \cite{carleb17,carleb18}.
All these algorithms use time polynomial in the input degree.

We consider the following counting problem: given $f \in F[x]$ and a
divisor $d$ of its degree, how many $(g,h)$ are there with $f = g
\circ h$ and $\deg g = d$? Under a suitable normalization, the answer
in the tame case is simple: at most one. However, we address this
question for additive polynomials, in some sense an ``extremely wild''
case, and determine both the structure and the number of such
decompositions. This involves three steps:
\begin{itemize}
\item a bijective correspondence between decompositions of an additive
  polynomial $f$ and Frobenius-invariant subspaces of its root space
  $V_{f}$ in an algebraic closure of $F$ (\autoref{sec:InNo}),
\item a description of the $A$-invariant subspaces of an $F$-vector space
  for a  matrix $A \in F^{\nxn}$ in rational Jordan form
  (\autoref{sec:rati-norm-forms}), and
\item an efficient algorithm to compute the rational Jordan form of the
  Frobenius automorphism on $V_{f}$ (\autoref{sec:Frob}). Its runtime
  is polynomial in $\log_{p}(\deg f)$.
\end{itemize}
A combinatorial result of \cite{fri11} counts the relevant
Frobenius-invariant subspaces of $V_{f}$ and thus our decompositions
(\autoref{sec:count-invar-subsp}). We also count the number of maximal
chains of Frobenius-invariant subspaces and thus the complete
decompositions. Our algorithm deals with squarefree polynomials, and
we give a reduction for the general case (\autoref{sec:general}).

Some of the results in the present paper are described in an Extended
Abstract \citep{gatgie09}.  A version of the present paper is
available at \url{https://arxiv.org/abs/1005.1087}. Implementations of
all algorithms in \citetalias{sage611} are available at
\url{https://github.com/zieglerk/polynomial_decomposition}.

\section{Additive polynomials and vector spaces}
\label{sec:InNo}

\todo{INFO this section is coordinate-free; $q$ and $r$ independent powers of
$p$}

Additive (or linearized) polynomials have a rich mathematical
structure. Introduced by \cite{ore33b}, they play an important role in
the theory of finite and function fields and have found many
applications in coding theory and cryptography.  See
\citet[Section~3.4]{lidnie97} for an introduction and
survey. In this section, we establish connections between
components of additive polynomials, subspaces of root spaces, and
factors of so-called projective polynomials.

We focus on additive polynomials over finite fields $\F$, though some of
these results hold more generally for any field of
characteristic $p > 0$. Let $r$ be a power of $p$ and let
\begin{equation}
 \F [x; r] = \left\{ \sum_{0\leq i\leq n} a_i x^{r^i} \colon n\in\ZZ_{\geq
   0}, \quad a_0,\ldots,a_n\in \F \right\}
\end{equation}
be the set of \emph{$r$-additive} (or \emph{$r$-linearized})
polynomials over $\F$. For $\F = \Fr$, we fix an algebraic closure $\Fbar \supseteq
\Fr$. Then these are the polynomials $f$ such that $f(a \alpha + b \beta
) = a f(\alpha) + b f(\beta)$ for any $a, b\in \Fr$ and
$\alpha, \beta \in \Fbar$. The $r$-additive polynomials form a non-commutative
ring under the usual addition and composition. It is a principal left
(and right) ideal ring with a left (and right) Euclidean algorithm;
see \citet[Chapter~1, Theorem~1]{ore33b}. For $f, h \in \pP{\F}{r}$,
we find
\begin{equation}
  \label{eq:10}
  h \text{ is a factor of } f \Longleftrightarrow h \text{ is a right
    component of } f
\end{equation}
after comparing division with remainder of $f$ by $h$ (in $\F[x]$) and
decomposition with remainder of $f$ by $h$ (in $\pP{\F}{r}$).
All components of an $r$-additive polynomial are $p$-additive, see
\citet[Theorem~4]{dorwha74} and \citet[Theorem~3.3]{gie88b}.

An additive polynomial is squarefree if its derivative is nonzero,
meaning that its linear coefficient $a_{0}$ is nonzero. To understand
the decomposition behavior of additive polynomials, it is sufficient
to restrict ourselves to monic squarefree elements of $\F [x; r]$. The
general (monic non-squarefree) case is discussed in
\autoref{sec:general}. For $f \in \F [x; r]$ with $\deg f = r^{n}$, we
call $n$ the \emph{exponent} of $f$, denote it by $\expn f$, and write
for $n \geq 0$
\begin{equation}
  \label{eq:8}
  \pP[n]{\F}{r} = \left\{f \in \pP{\F}{r} \colon f \text{ is monic
    squarefree with exponent } n\right\}.
\end{equation}

For $f \in \pP[n]{\F}{r}$, the set $V_{f}$ of all roots of $f$ in an
algebraic closure $\Fbar$ of $F$ forms an $\Fr$-vector space of
dimension $n$. From now on, we assume $q$ to be a power of $r$,
and let $\F = \Fq$ be a finite field with $q$
elements. Then $V_{f}$ is invariant under the $q$th power
\emph{Frobenius automorphism} $\sigma_{q}$, since for $\alpha \in \Fbar$
with $f(\alpha)=0$ we have $f(\sigma_q(\alpha)) = f(\alpha^{q})=f(\alpha)^{q}=0$, thus $\sigma_{q}(V_{f})
\subseteq V_{f}$, and $\sigma_{q}$ is injective. For $n \geq 0$, we
define
\begin{align}
\vV[n]{q}{\Fr} & = \left\{n\text{-dimensional } \sigma_{q}\text{-invariant }
  \Fr\text{-linear subspaces of } \Fqbar\right\}, \\
\psi_{n} & \colon\hspace*{5pt}
\begin{aligned}
  \pP[n]{\Fq}{r} & \to \vV[n]{q}{\Fr}, \\
  f & \mapsto V_{f} = \{\alpha \in \Fbar \colon f(\alpha)= 0\}.
\end{aligned}\label{eq:6}
\end{align}

Conversely, for any $n$-dimensional $\Fr$-vector space $V \subseteq
\Fbar$, the lowest degree monic polynomial $f_{V} = \prod_{\alpha \in
  V} (x-\alpha) \in \Fbar[x]$ with $V$ as its roots is a squarefree
$r$-additive polynomial of exponent $n$, see
\citet[Theorem~8]{ore33b}.  If $V$ is invariant under $\sigma_{q}$,
then $f_{V} \in \pP[n]{\Fq}{r}$. For $n \geq 0$, we define
\begin{equation}
\label{eq:7}
\varphi_{n} \colon\hspace*{5pt}
\begin{aligned}
  \vV[n]{q}{\Fr} & \to \pP[n]{\Fq}{r}, \\
  V & \mapsto f_{V} = \prod_{\alpha \in V} (x-\alpha).
\end{aligned}
\end{equation}

\citet[Chapter~1, \S\S~3--4]{ore33b} gives a correspondence between
monic squarefree $p$-additive polynomials and $\Fp$-vector spaces
which generalizes as follows.
\begin{proposition}
  \label{pro:1}
  For $r$ a power of a prime $p$,  $q$ a power of $r$, and $n \geq 0$, the maps
  $\psi_{n}$ and $\varphi_{n}$ are inverse bijections.
\end{proposition}

\subsection{Right components and invariant subspaces}
\label{sec:right-comp-invar}

The following refinement of \autoref{pro:1} is a cornerstone of this
paper. It provides a bijection between right components of a monic
original $f \in \pP[n]{\Fq}{r}$ and $\sigma_{q}$-invariant subspaces
of its root space $V_{f} \in \vV[n]{q}{\Fr}$. The latter are analyzed
with methods from linear algebra in \autoref{sec:rati-norm-forms}.
Those insights are then reflected back to questions about
decompositions, providing results that seem hard to obtain directly.

For $n \geq d \geq 0$, $f \in \pP[n]{\Fq}{r}$, and $V \in \vV[n]{q}{\Fr}$, we define
\begin{align}
H_{d}(f) = H_{q,r,d}(f) &  = \left\{\text{right components } h \in \pP[d]{\Fq}{r} \text{
  of } f\right\} \subseteq \pP[d]{\Fq}{r}, \\
L_{d}(V) = L_{q,r,d}(V) & = \{d\text{-dimensional }
\sigma_{q}\text{-invariant } \Fr\text{-linear subspaces of } V\} \\
& \subseteq \vV[d]{q}{\Fr},
\end{align}
where we omit $q$ and $r$ from the notation when they are clear from
the context.
We also set $H_{q,r,d}(f) = L_{q,r,d}(V) = \varnothing$ for $d<0$.

\begin{proposition}
  \label{pro:squarefree}
  Let $n \geq d \geq 0$, $r$ be a power of a prime $p$,  $q$ a power of $r$, and $f \in \pP[n]{\Fq}{r}$.
  Then the restrictions of $\psi_{d}$ and $\varphi_{d}$ are inverse
  bijections between $H_{q,r,d}(f)$ and $L_{q,r,d}(V_{f})$.
\end{proposition}

\begin{proof}

  For $h \in H_{d}(f)$, we have $h \mid f$ by \eqref{eq:10}, and thus
  $V_{h} \subseteq V_{f}$. Since $h \in \Fq[x;r]_{d}$, we have $\dim
  V_{h} = d$ and $V_{h} \in L_{d}(V_{f})$. Conversely, for $W \in
  L_{d}(f)$, we have $W \subseteq V_{f}$ and $f_{W}$ is a squarefree
  divisor of $f$ with $\expn f_{W} = d$. From \eqref{eq:10}, we have
  $f_{W} \in H_{d}(f)$. Thus, $\psi_{d}(\varphi_{d}(L_{d}(f)))
  \subseteq L_{d}(f)$ and $\varphi_{d}(\psi_{d}(H_{d}(f))) \subseteq
  H_{d}(f)$. Since both sets are finite and both maps are injective,
  we have equalities and the claim follows.
\end{proof}

Thus, under the conditions of \autoref{pro:squarefree}, we have for $h \in
\Fq[x;r]_{d}$
\begin{equation}
  \label{eq:15}
  h \mid f \Longleftrightarrow V_{h} \subseteq V_{f}
  \Longleftrightarrow h \in H_{d}(f),
\end{equation}
as an extension of \eqref{eq:10}.

\subsection{General additive polynomials}
\label{sec:general}

We generalize \autoref{pro:squarefree} from squarefree to all monic additive
polynomials. We can write \emph{any} monic $\bar{f} \in \pP{\F}{r}$ as $g
\circ x^{r^{m}}$ with unique $m \geq 0$ and unique monic squarefree
$g \in \pP{\F}{r}$. Then
\begin{equation}
  \label{eq:22}
  \bar{f} = g \circ x^{r^{m}} = x^{r^{m}} \circ f
\end{equation}
with unique monic squarefree $f \in \pP{\F}{r}$ and the coefficients
of $f$ are the $r^{m}$th roots of the coefficients of $g$, see
\citet[Section~3]{gie88b}. Composing an additive polynomial with
$x^{r^{m}}$ from the left leaves the root space invariant and we have
\begin{equation}
  \label{eq:37}
  V_{\bar{f}} = V_{x^{r^{m}} \circ f} = V_{f}.
\end{equation}
We now relate the right components of $\bar{f}$ to the right components of $f$.

\begin{proposition}
  \label{thm:equivalent}
  Let $m, n \geq 0$, $m + n \geq d \geq 0$, $r$ be a power of a
  prime $p$ and $q$ a power of $r$, $0 \leq d \leq m+n$, and $f \in \pP[n]{\Fq}{r}$. For monic $\bar{f} =
  x^{r^{m}} \circ f \in \pP{\Fq}{r}$ with exponent $m + n$, we have a
  bijection between any two of the following three sets:
  \begin{ronumerate}
  \item\label{it:eq:1} $\{ \text{monic right components } \bar{h} \in
    \pP{\Fq}{r} \text{ of } \bar{f} \text{ with exponent }d\}$,
  \item\label{it:eq:2} the union of all $H_{i}(f)$ for $d - m \leq i \leq d$, and
  \item\label{it:eq:3} the union of all $L_{i}(V_{f})$ for $d - m \leq i \leq d$.
  \end{ronumerate}
\end{proposition}

\begin{proof}
  We begin with a bijection between \ref{it:eq:1} and
  \ref{it:eq:2}. Following \eqref{eq:22}, we can write every
  $\bar{h}$ in \ref{it:eq:1} as $x^{r^{d-i}} \circ h$ with unique $i$ satisfying $d-m
  \leq i \leq d$ and unique monic squarefree $h \in
  \Fq[x;r]_{i}$. Then $V_{h} \subseteq V_{\bar{f}} = V_{f}$ and $h \in
  H_{i}(f)$ by \eqref{eq:15}. Conversely, let $d - m \leq i \leq d$
  and $h \in H_{i}(f)$. Then $f = g \circ h$ for some $g \in \Fq[x;
  r]_{n-i}$ and $\bar{f} = x^{r^{m-d+i}} \circ \tilde{g} \circ
  x^{r^{d-i}} \circ h$, where the coefficients of $\tilde{g}$ are the
  $r^{d-i}$th roots of the coefficients of $g$. Thus $\bar{h} =
  x^{r^{d-i}} \circ h$ is a monic right component of $\bar{f}$ with
  exponent $d$. Together this yields a one-to-one correspondence between
  \ref{it:eq:1} and \ref{it:eq:2}.

  \autoref{pro:squarefree} provides a bijection between \ref{it:eq:2} and \ref{it:eq:3}.

\end{proof}

We note that for $d>n$, all three sets are empty.

\subsection{Projective and subadditive polynomials}
\label{sec:proj-subadd}

As an aside, we exhibit two further sets of polynomials that are in
bijective correspondence with $H_{d}(f)$; this will not be used beyond
this subsection, but illustrates the wide range of applications. Let
$f = \sum_{0 \leq i \leq n} a_{i} x^{r^{i}} \in F[x; r]$ and $t$ be a
positive divisor of $r-1$. We have $f = x \cdot (\pi_{t}(f) \circ
x^{t})$ for $\pi_{t}(f) = \sum_{0 \leq i \leq n} a_{i}
x^{(r^{i}-1)/t}$. \cite{abh97} introduced the \emph{projective
  polynomials}
\begin{equation} \label{eq:34}
  \pi_{r-1}(x^{r^{n}} + a_{1}x^{r} + a_{0}x) = x^{(r^{n}-1)/(r-1)}+a_{1}x+a_{0},
\end{equation}
which may have, over function fields of positive characteristic, nice
Galois groups such as projective general or projective special linear
groups. Projective polynomials appear naturally in coding theory
(e.g., \cite{helkho08a}, \cite{zenli08}) and the study of difference
sets (e.g., \cite{dil02}, \cite{blu03}). They can be used to construct
strong Davenport pairs explicitly \citep{blu04b} and determine whether
a quartic power series is actually hyperquadratic
\citep{blulas06}. The linear shifts of \eqref{eq:34} are closely
related to group actions on irreducible polynomials over $\Fq$
\citep{stitop12}. The cardinality of the value set of a (possibly
non-additive) polynomial $f \in \Fq[x]$ is determined by the maximal
$s$, $t$ such that $f = x^{s} \cdot (\bar{f} \circ x^{t})$ for some
$\bar{f} \in \Fq[x]$ \citep{akbghi09}. \cite{blu04a} shows that
\eqref{eq:34} has exactly $0$, $1$, $2$, or $r+1$ roots in $\Fq$ for
$q$ a power of $r$. \cite{helkho10} count the roots for $q$ and $r$
independent powers of $2$.

The polynomial
\begin{equation}
  \label{eq:1}
  \rho_{t}(f) = x \cdot (x^{t} \circ \pi_{t}(f)) = x \cdot (\pi_{t}(f))^{t}
\end{equation}
is called \emph{$(r,t)$-subadditive} (or simply
\emph{subadditive}). We have $\rho_{t}(f) \circ x^{t} = x^{t} \circ f$
and in particular $\rho_{1}(f) = f$. Subadditive polynomials were
introduced by \cite{coh90c} to study their role as permutation
polynomials. \cite{henmat99} connect their decomposition behavior to
that of additive polynomials and provide the bijection between
\ref{it:1} and \ref{it:2} in the following
proposition. \cite{couhav04} use this connection to apply
\citeauthor{odo99}'s (\citeyear{odo99}) counting formula for
$p$-additive polynomials and \citeauthor{gie98}'s (\citeyear{gie98})
decomposition algorithm for additive polynomials to subadditive
polynomials.

\begin{proposition}
  \label{pro:2}
  Let $n \geq d \geq 0$, $r$ be a power of a prime $p$,  $q$ a power of $r$,
  $t$ a positive divisor of $r-1$,  and $f \in \Fq
  [x; r]_{n}$. Then we have bijections between any two of the
  following three sets.
  \begin{ronumerate}
  \item\label{it:1} $H_{d}(f)$,
  \item\label{it:3} the set of monic factors of $\pi_{t}(f)$ that are
    of the form $\pi_{t}(h)$ for some $h \in F[x; r]_{d}$, and
  \item\label{it:2} the set of monic $(r,t)$-subadditive right
    components of $\rho_{t}(f)$ of degree $r^{d}$.
  \end{ronumerate}
  In particular, the maps $\pi_{t}$ and $\rho_{t}$ are bijections
  from \ref{it:1} to \ref{it:3} and to \ref{it:2}, respectively.
\end{proposition}

\begin{proof}
For the bijection between \ref{it:1} and \ref{it:3}, it is sufficient to show that the following
statements are equivalent for $h \in \Fq[x;r]_{d}$:
  \begin{itemize}
  \item $h$ is a right component of $f$;
  \item $h = x \cdot (\pi_{t}(h) \circ x^{t})$ is a factor of $f = x
    \cdot (\pi_{t}(f) \circ x^{t})$;
  \item $\pi_{t}(h)$ is a factor of $\pi_{t}(f)$.
  \end{itemize}
The first two items are equivalent by \eqref{eq:10}, and so are the last two
since $\pi_{t}(h)\pi_{t}(f)$ is coprime to $x$ for squarefree $h$ and $f$.

The bijection between \ref{it:1} and \ref{it:2} is due to \citet[Theorem~4.1]{henmat99}.
\end{proof}

Irreducible factors in \ref{it:3} correspond to components in
\ref{it:1} and \ref{it:2} that are indecomposable over $\Fq[x;r]$ and
$\rho_{t}(\Fq[x;r])$, respectively. For $d = 1$ and $t = r - 1$, this
yields the following criterion by \citeauthor{ore33b}.

\begin{fact}[{\citealt[Theorem~3]{ore33b}}]
  \label{cor:2}
  For $n$, $r$, and $F$ as in \autoref{pro:2}, $f \in \Fq[x;r]_{n}$
  and $a \in \F^{\times}$, we have
  \begin{equation}
    \label{eq:14}
x^{r}-ax \in H_{1}(f) \Longleftrightarrow \pi_{r-1}(f)(a) = 0.
  \end{equation}
\end{fact}

\section{The rational Jordan form}
\label{sec:rati-norm-forms}

\todo{INFO this section only $r$, no $q$!}
\newcommand{\frnew}{\Fr}
\renewcommand{\frnew}{\F}

The usual Jordan (normal) form of a matrix contains the
eigenvalues. It is unique up to permutations of the Jordan blocks. The
rational Jordan form of a matrix is a generalization, with eigenvalues
in a proper extension of the ground field being represented by the
companion matrix of their minimal polynomial. Forms akin to the
rational Jordan form were investigated already by \cite{fro11} and the
underlying decomposition of the vector space is described by
\citet[Chapter~VII]{gan59}. A detailed discussion of rational normal
forms can be found in \citet[Chapter~6]{lun87}.

Let $A$ be a square matrix with entries in
$\frnew$. We factor the \emph{minimal polynomial} of $A$ over $\frnew$
completely and obtain
$\minpoly(A) = u_{1}^{k_{1}} \cdots u_{t}^{k_{t}} \in \frnew[y]$ with $t$
pairwise distinct monic irreducible $u_{i} \in \frnew[y]$ and $k_{i} > 0$
for $1 \leq i \leq t$. We call $u_{i}$ an \emph{eigenfactor} of $A$
and $\ker(u_{i}(A))$ its \emph{eigenspace}. 

For any $u = \sum_{0 \leq i \leq m} a_{i} y^{i} \in \frnew[y]$ with
$a_{m} = 1$, we have the \emph{companion matrix}
\begin{equation}
  \label{eq:3}
  C_{u} = \begin{pmatrix}
0 &        &        &   & -a_{0} \\
1 & \ddots &        &   & -a_{1} \\
0 & \ddots & \ddots &   & \vdots \\
  & \ddots & \ddots & 0 & \vdots \\
  &        & 0      & 1 & -a_{m-1}
  \end{pmatrix}
\in \frnew^{m \times m}
\end{equation}
with $\minpoly(C_{u}) = u$. The \emph{rational Jordan block} of
\emph{order} $\ell > 0$ for $u$ is
\begin{equation}
  \label{eq:4}
  J_{u}^{(\ell)} = \begin{pmatrix}
C_{u} & I_{m}       &        &  \\
     & C_{u} & \ddots       &  \\
  & & \ddots & I_{m} \\
  &        &       & C_{u}
  \end{pmatrix}
\in \frnew^{(\ell m) \times (\ell m)},
\end{equation}
where $I_{m}$ is the $m \times m$ identity matrix. For linear $u = y -
a \in \frnew[y]$, we have $C_{u} = (a)$ and the rational Jordan
blocks are the Jordan blocks of the usual \emph{Jordan form}. The
arrangement of rational Jordan blocks along the main diagonal gives a
rational Jordan form.

\begin{definition}
  A \emph{rational Jordan matrix} over $\frnew$ is a matrix of the shape
  \begin{equation}
    \label{eq:jordanform}
    A = \diag( J_{u_{1}}^{(\ell_{11})}, \dots,
    J_{u_{1}}^{(\ell_{1s_{1}})}, \dots, J_{u_{t}}^{(\ell_{t1})},
    \dots, J_{u_{t}}^{(\ell_{ts_{t}})})
  \end{equation}
  with $t\geq 1$,  pairwise distinct monic irreducible $u_{1}, \dots, u_{t} \in
  \frnew[y]$, $s_i \geq 1$, and $\ell_{i1} \geq \ell_{i2} \geq \cdots \geq \ell_{is_i}$ for $1 \leq i \leq t$.
  \end{definition}

\citet[Lemma~8.1]{gie95} shows that $\minpoly(J_{u}^{(\ell)}) = u^{\ell}$,
and thus $\minpoly(A)$ $= u_{1}^{\ell_{i1}} \cdots u_{t}^{\ell_{t1}}$.
 Every matrix over $\frnew$ is similar to a rational
Jordan matrix over $\frnew$, see, e.g., \citet[Theorem~8.3]{gie95}, which we
call the \emph{rational Jordan form} of the matrix. The
eigenvalues and their multiplicities are preserved by this similarity
transformation and the rational Jordan form is unique up to
permutation of the rational Jordan
blocks. \citet[Corollary~8.6]{gie95} shows how to transform an
$\nxn$ matrix over $\frnew$ into rational Jordan form using
$\softOh{n^{\omega} + n \log r}$ field operations, where $\omega$ is
the exponent of square matrix multiplication over $\frnew$. 
This matches the lower bound
$\Omega(n^{\omega})$ for this problem up to polylogarithmic factors.
The ``textbook'' method gives $\omega \leq 3$ and \cite{gal14} shows
$\omega < 2.3728639$.

We extract the purely combinatorial data from a rational Jordan form
$A \in \frnew^{\nxn}$ as in \eqref{eq:jordanform}. For $1 \leq i \leq t$
and $1 \leq j \leq \ell_{i1}$, let $\lambda_{ij}$ denote the number of
rational Jordan blocks of order $j$ for the eigenfactor $u_{i}$. The
formulae for $\lambda_{ij}$ over the algebraic closure, see, e.g.,
\citet[p.~155]{gan59}, generalize as
\begin{align}
  \lambda_{ij} \cdot \deg u_{i} & = \rk(u_{i}^{j-1}(A)) -
  2\rk(u_{i}^{j}(A)) + \rk(u_{i}^{j+1}(A)) \\
  & = 2 \nul(u_{i}^{j}(A)) - \nul(u_{i}^{j-1}(A)) -
  \nul(u_{i}^{j+1}(A)), \label{eq:19}
\end{align}
where $u_{i}^{0}(A) = I_{n}$ and $\nul B = n - \rk
B$ is  the \emph{nullity} of $B$ for any $B \in \frnew^{\nxn}$. The
vector
$\lambda (u_{i}) =(\deg u_{i}; \lambda_{i1}, \lambda_{i2}, \dots, \lambda_{i\ell_{i1}})$
of positive integers is the \emph{species} of $u_{i}$ (in $A$). This abstracts away the
arrangement of the rational Jordan blocks as well as the actual
factors $u_{i}$. The multiset of all the species of eigenfactors in
$A$ is then called the \emph{species $\lambda(A)$ of $A$}. This notion was
introduced by \cite{kun81} over the algebraic closure and generalized
to finite fields by \cite{fri11}.

\autoref{tab:dim-3} gives all similarity classes of rational Jordan
forms $A$ in $\frnew^{3\times 3}$ and their species. The notation
$3 \times (1;1)$ indicates that the species $(1;1)$ occurs three times
in the multiset. We also list, for every species, the lattice
$\mathcal{L}(A)$ of $A$-invariant subspaces in a $3$-dimensional
$\frnew$-vector space, the number $\# L_{1} (A)$ of $1$-dimensional
$A$-invariant subspaces, and the number $\numchains{A}$ of maximal
$A$-invariant subspace chains \eqref{eq:25}.

In the next subsection, we derive the latter from the species. In
\autoref{sec:Frob}, we show how to compute the rational Jordan form of
the Frobenius automorphism on the root space of an additive polynomial
\emph{without} the (costly) computation of a basis.

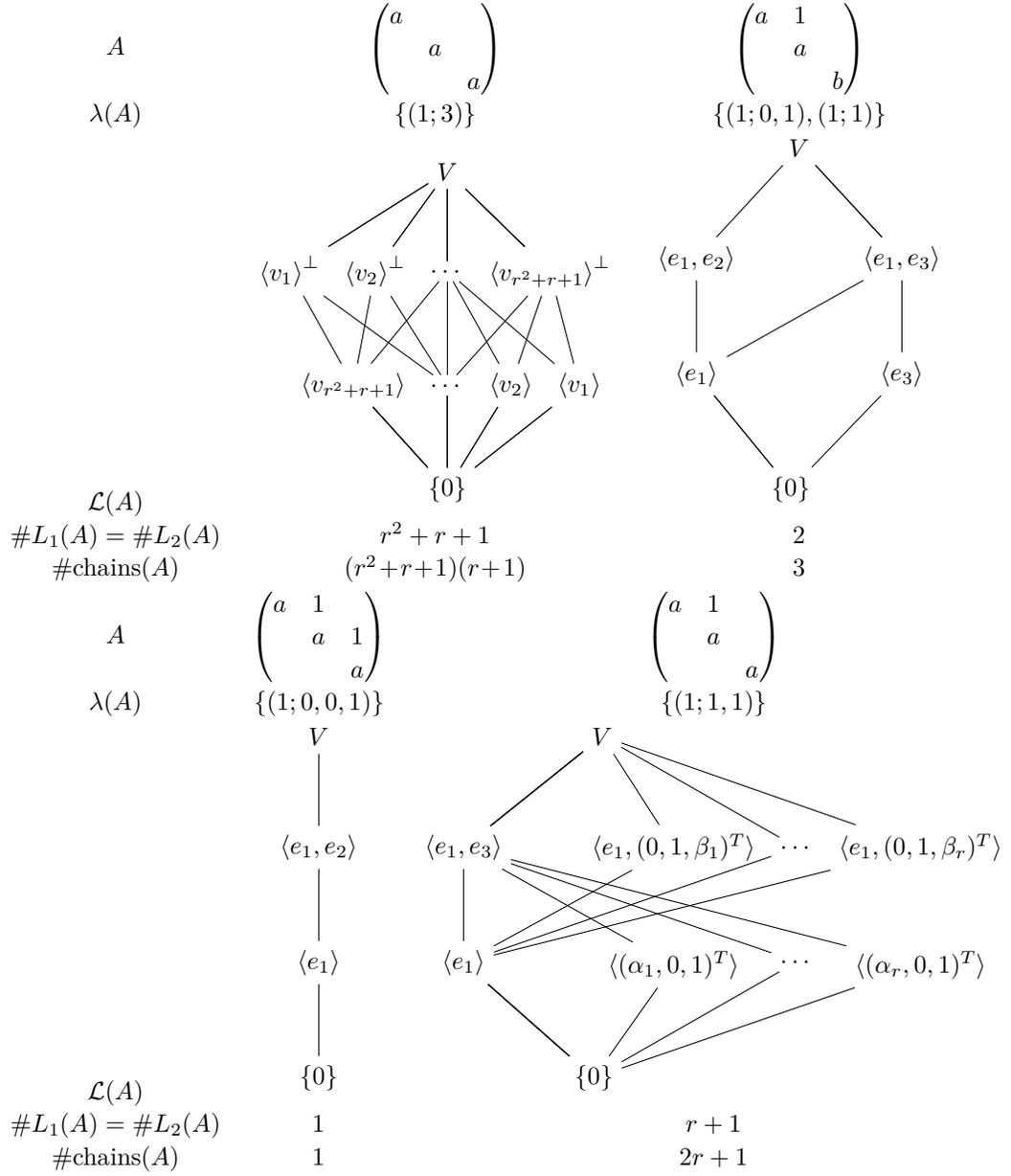
\begin{table}
  \centering
  \begin{tabular}{CCC}
  A &
      \begin{pmatrix}
        a & & \\
        & a & \\
        & & a
      \end{pmatrix} &
      \begin{pmatrix}
        a & 1 & \\
        & a & \\
        & & b
      \end{pmatrix} \\

  \lambda(A) &
               \{ (1; 3) \} & \{(1;0,1), (1;1)\} \\

  \mathcal{L}(A) &
\begin{tikzpicture}
\node(V)                          {$V$};
\node(bdots) [below=1cm of V] {$\dots$};
\node(v2o) [left=.1cm of bdots] {$\braket{v_{2}}^{\bot}$};
\node(v1o) [left=.1cm of v2o] {$\braket{v_{1}}^{\bot}$};
\node(vro) [right=.1cm of bdots] {$\braket{v_{r^{2}+r+1}}^{\bot}$};
\node(adots)   [below=1.3cm of bdots] {$\dots$};
\node(vr)   [left=.1cm of adots] {$\braket{v_{r^{2}+r+1}}$};
\node(v2)   [right=.1cm of adots] {$\braket{v_{2}}$};
\node(v1)   [right=.1cm of v2] {$\braket{v_{1}}$};
\node(0)    [below=1cm of adots]     {$\{0\}$};

\draw(V) -- (v1o) -- (vr) -- (0);
\draw(V) -- (v1o) -- (adots) -- (0);
\draw(V) -- (v2o) -- (vr) -- (0);
\draw(V) -- (v2o) -- (adots) -- (0);
\draw(V) -- (bdots) -- (vr) -- (0);
\draw(V) -- (bdots) -- (adots) -- (0);
\draw(V) -- (bdots) -- (v2) -- (0);
\draw(V) -- (bdots) -- (v1) -- (0);
\draw(V) -- (vro) -- (adots) -- (0);
\draw(V) -- (vro) -- (v2) -- (0);
\draw(V) -- (vro) -- (v1) -- (0);

\end{tikzpicture}
          &
                      \begin{tikzpicture}
\node(V)                          {$V$};
\node(e1e2) [below left=1cm and 0.5cm of V] {$\braket{e_{1}, e_{2}}$};
\node(e1e3) [below right=1cm and 0.5cm of V]  {$\braket{e_{1}, e_{3}}$};
\node(e1)   [below=1cm of e1e2]       {$\braket{e_{1}}$};
\node(e3)   [below=1cm of e1e3]       {$\braket{e_{3}}$};
\node(0)    [below right=1cm and 0.5cm of e1]     {$\{0\}$};

\draw(V) -- (e1e2) -- (e1) -- (0);
\draw(V) -- (e1e3) -- (e1) -- (0);
\draw(V) -- (e1e3) -- (e3) -- (0);
\end{tikzpicture}
\\
    \# L_{1}(A) = \# L_{2}(A) & r^2+r+1 & 2 \\
  \numchains{A} &
  (r^2\mskip-3mu+\mskip-3mu{}r\mskip-3mu+\mskip-3mu1)(r\mskip-3mu+\mskip-2mu1) & 3\\
  \end{tabular}
  \begin{tabular}{CCC}
  A &
      \begin{pmatrix}
        a & 1 & \\
        & a & 1 \\
        & & a
      \end{pmatrix} &
      \begin{pmatrix}
        a & 1 & \\
        & a & \\
        & & a
      \end{pmatrix} \\

  \lambda(A) &
               \{(1;0,0,1)\}  &\{(1;1,1)\} \\

  \mathcal{L}(A) & \begin{tikzpicture}
\node(V)                          {$V$};
\node(e1e2) [below=1cm of V] {$\braket{e_{1}, e_{2}}$};
\node(e1)   [below=1cm of e1e2] {$\braket{e_{1}}$};
\node(0)    [below=1cm of e1]     {$\{0\}$};

\draw(V) -- (e1e2) -- (e1) -- (0);
\end{tikzpicture}
        &
              \begin{tikzpicture}
\node(V)                          {$V$};
\node(e1e3) [below left=1cm and 1cm of V] {$\braket{e_{1}, e_{3}}$};
\node(e1b1) [right=1cm of e1e3] {$\braket{e_{1}, (0, 1, \beta_{1})^{T}}$};
\node(bdots) [right=.1cm of e1b1] {$\dots$};
\node(e1br) [right=.1cm of bdots] {$\braket{e_{1}, (0, 1, \beta_{r})^{T}}$};
\node(e1)   [below=1cm of e1e3] {$\braket{e_{1}}$};
\node(a1)   [below=1cm of e1b1] {$\braket{(\alpha_{1}, 0, 1)^{T}}$};
\node(adots)   [below=1.3cm of bdots] {$\dots$};
\node(ar)   [below=1cm of e1br] {$\braket{(\alpha_{r}, 0, 1)^{T}}$};
\node(0)    [below right=1cm and 1cm of e1]     {$\{0\}$};

\draw(V) -- (e1e3) -- (e1) -- (0);
\draw(V) -- (e1e3) -- (a1) -- (0);
\draw(V) -- (e1e3) -- (adots) -- (0);
\draw(V) -- (e1e3) -- (ar) -- (0);
\draw(V) -- (e1b1) -- (e1) -- (0);
\draw(V) -- (bdots) -- (e1) -- (0);
\draw(V) -- (e1br) -- (e1) -- (0);

\end{tikzpicture}
\\
    \# L_{1}(A) = \# L_{2}(A) & 1 & r+1 \\
  \numchains{A} & 1 & 2r+1 \\
  \end{tabular}
\end{table}

\begin{table}
  \begin{tabular}{CCC}
  A &
      \begin{pmatrix}
        a & & \\
        & a & \\
        & & b
      \end{pmatrix} &
      \begin{pmatrix}
        0 & & c_{0} \\
        1 & 0 & c_{1} \\
          & 1 & c_{2}
      \end{pmatrix} \\

  \lambda(A) &
  \{(1;2),(1;1)\} & \{(3;1)\} \\
    \mathcal{L}(A) &
\begin{tikzpicture}
\node(V)                          {$V$};
\node(bdots) [below=1cm of V] {$\dots$};
\node(a1e3) [left=0.1cm of bdots] {$\braket{(1, \alpha_{1}, 0)^{T}, e_{3}}$};
\node(e1e2) [left=0.1cm of a1e3] {$\braket{e_{1}, e_{2}}$};
\node(are3) [right=0.1cm of bdots] {$\braket{(1, \alpha_{r}, 0)^{T}, e_{3}}$};
\node(e2e3) [right=0.1cm of are3] {$\braket{e_{2}, e_{3}}$};
\node(adots) [below=1.4cm of bdots] {$\dots$};
\node(e2)   [below=1cm of e1e2] {$\braket{e_{2}}$};
\node(a1)   [below=1cm of a1e3] {$\braket{(1, \alpha_{1}, 0)^{T}}$};
\node(ar)   [below=1cm of are3] {$\braket{(1, \alpha_{r}, 0)^{T}}$};
\node(e3)   [below=1cm of e2e3] {$\braket{e_{3}}$};
\node(0)    [below=1cm of adots]     {$\{0\}$};

\draw(V) -- (e1e2) -- (e2) -- (0);
\draw(V) -- (e1e2) -- (a1) -- (0);
\draw(V) -- (e1e2) -- (adots) -- (0);
\draw(V) -- (e1e2) -- (ar) -- (0);

\draw(V) -- (a1e3) -- (a1) -- (0);
\draw(V) -- (a1e3) -- (e3) -- (0);
\draw(V) -- (bdots) -- (adots) -- (0);
\draw(V) -- (bdots) -- (e3) -- (0);
\draw(V) -- (are3) -- (ar) -- (0);
\draw(V) -- (are3) -- (e3) -- (0);

\draw(V) -- (e2e3) -- (e2) -- (0);
\draw(V) -- (e2e3) -- (e3) -- (0);
\end{tikzpicture}
                 &
            \begin{tikzpicture}
\node(V)                          {$V$};
\node(0)    [below=3.5cm of V]     {$\{0\}$};

\draw(V) -- (0);
\end{tikzpicture}
\\
    \# L_{1}(A) = \# L_{2}(A) & r+2 & 0 \\
  \numchains{A} &
  3(r+1) & 1
  \end{tabular}
  \begin{tabular}{CCC}
  A &
      \begin{pmatrix}
        a & \\
        & 0 & b_{0} \\
        & 1 & b_{1}
      \end{pmatrix} &
      \begin{pmatrix}
        a & & \\
        & b & \\
        & & c
      \end{pmatrix}\\

  \lambda(A) &
\{(1;1),(2;1)\}  & \{3 \times (1;1)\}  \\
    \mathcal{L}(A) &
          \begin{tikzpicture}
\node(V)                          {$V$};
\node(0)    [below=3.5cm of V]     {$\{0\}$};
\node(e1)   [above left=1cm and 0.5cm of 0]       {$\braket{e_{1}}$};
\node(e2e3) [below right=1cm and 0.5cm of V]  {$\braket{e_{2}, e_{3}}$};

\draw(V) -- (e1) -- (0);
\draw(V) -- (e2e3) -- (0);
\end{tikzpicture}
          &
\begin{tikzpicture}
\node(V)                          {$V$};
\node(e1e2) [below left=1cm and 1cm of V] {$\braket{e_{1}, e_{2}}$};
\node(e1e3) [below=1cm of V]  {$\braket{e_{1}, e_{3}}$};
\node(e2e3) [below right=1cm and 1cm of V] {$\braket{e_{2}, e_{3}}$};
\node(e1)   [below=1cm of e1e2]       {$\braket{e_{1}}$};
\node(e2)   [below=1cm of e1e3]       {$\braket{e_{2}}$};
\node(e3)   [below=1cm of e2e3]       {$\braket{e_{3}}$};
\node(0)    [below=1cm of e2]     {$\{0\}$};

\draw(V) -- (e1e2) -- (e1) -- (0);
\draw(V) -- (e1e2) -- (e2) -- (0);
\draw(V) -- (e1e3) -- (e1) -- (0);
\draw(V) -- (e1e3) -- (e3) -- (0);
\draw(V) -- (e2e3) -- (e2) -- (0);
\draw(V) -- (e2e3) -- (e3) -- (0);
\end{tikzpicture}
\\
    \# L_{1}(A) = \# L_{2}(A) & 1  & 3  \\
  \numchains{A} & 2 & 6
  \end{tabular}
  \caption{All similarity classes of rational Jordan forms $A \in
    \frnew^{3\times 3}$, where $a, b, c \in \frnew$ are pairwise distinct
    eigenvalues and the eigenfactors $ y^{2} - b_{1}y - b_{0}$ and $ y^{3} - c_{2}
    y^{2} - c_{1} y - c_{0}$ are irreducible over $\frnew$.}
          \label{tab:dim-3}
        \end{table}

\subsection{The number of invariant subspaces}
\label{sec:count-invar-subsp}

\todo{INFO in this subsection, $d$ is the dimension of a subspace; in
  other sections it is the degree of the field extension; these
  notions never ``meet''.}

Let $r$ be a power of the prime $p$ and 
$A \in \Fr^{\nxn}$ be a rational Jordan matrix as in
\eqref{eq:jordanform} with $\minpoly(A) = u_{1}^{k_{1}} \cdots
u_{t}^{k_{t}}$, where  $u_{1}, \ldots, u_t \in \Fr[y]$ are pairwise distinct
monic irreducible, and $k_{i} > 0$ for $1 \leq i \leq t$. $A$ operates on every
$n$-dimensional $\Fr$-vector space $V$ and we have the corresponding
\emph{primary vector space decomposition}
\begin{equation}
  \label{eq:9}
  V = V_{1} \oplus V_{2} \oplus \dots \oplus V_{t},
\end{equation}
where $V_{i} = \ker(u_{i}^{k_{i}}(A))$ is the \emph{generalized
  eigenspace} of $u_{i}$ for $1 \leq i \leq t$.

We ask two counting questions, motivated by the connection to
decomposition.
\begin{ronumerate}
\item\label{it:q1} What is the number $\# L_{d}(A)$ of $d$-dimensional
  $A$-invariant subspaces of $V$ for a given $d$?
\item\label{it:q2} What is the number $\numchains{A}$ of
  maximal chains
  \begin{equation}
    \label{eq:25}
    \{0\} = U_{0} \subsetneq U_{1} \subsetneq \dots \subsetneq U_{e} = V
  \end{equation}
  of $A$-invariant subspaces $U_{j}$ for $0 \leq j \leq e$, where $e$
  is the Krull dimension of $V$?
\end{ronumerate}
The $A$-invariant subspaces of $V$ constitute the complete lattice
$\mathcal{L}(A)$ with minimum $\{0\}$ and maximum $V$. In this
lattice's Hasse diagrams, question~\ref{it:q1} asks for the number of
nodes of a given dimension and question~\ref{it:q2} asks for the
number of paths from the minimum to the maximum.

First, we discuss question~\ref{it:q1}. Let $\gf(A) = \sum_{0 \leq d
  \leq n} \gf_{d}z^{d} \in \ZZ_{\geq 0}[z]$ be the \emph{generating
  function} for the number $\gf_{d} = \# L_{d}(A)$ of $d$-dimensional
$A$-invariant subspaces of $V$. The $A$-invariant subspace lattice
$\mathcal{L}(A)$ is self-dual, see \citet[Theorem~3]{brifil67}, and
thus the generating function is symmetric with $\gf_{d} = \gf_{n-d}$
for all $0 \leq d \leq n$.

Let $A_{i}$ denote the restriction of $A$ to $V_{i}$ as in
\eqref{eq:9}, and $\mathcal{L}(A_{i})$ and $\gf(A_{i})$ be the lattice
and generating function of the $A_{i}$-invariant subspaces of $V_{i}$,
respectively. \citet[Theorem~1]{brifil67} show that
\begin{equation}
  \label{eq:24}
\mathcal{L}(A) = \prod_{1 \leq i \leq t} \mathcal{L}(A_{i}) \text{ and
  thus } \gf (A) = \prod_{1 \leq i \leq t} \gf (A_{i}).
\end{equation}
Thus it suffices to study \emph{$A$-primary vector spaces}, where
$\minpoly(A) = u^{k}$ is the $k$th power of an irreducible polynomial
$u$ of some degree $m$. If an $n$-dimensional $A$-primary vector space has species
$\lambda(A) = \{(m, \lambda_{1}, \lambda_{2}, \dots, \lambda_{k})\}$,
then there is a rational Jordan form $B \in \Fr^{n/m \times n/m}$ with
species $\lambda(B) = \{(1, \lambda_{1}, \lambda_{2}, \dots,
\lambda_{k})\}$ and
\begin{equation}
  \label{eq:27}
  \mathcal{L}(A) \cong \mathcal{L}(B) \text{ and } \gf(A) = \gf(B)
  \circ z^{m}.
\end{equation}
It is therefore enough to study $A$-primary vector spaces, where
$\minpoly(A)$ is the power of a linear polynomial. In this situation,
we now compute $\gf_{1}(A)$.

From the theory of $q$-series, we use the
\emph{$q$-bracket} (also \emph{$q$-number})
\begin{equation}
  \label{eq:28}
  [n]_{q} = \frac{q^{n} - 1}{q - 1}
\end{equation}
of an integer $n$.

\begin{lemma}
  \label{pro:primary}
  Let $A \in \Fr^{\nxn}$ be a rational Jordan form as in
  \eqref{eq:jordanform} with $\minpoly(A) = u^{k}$ for some linear
  $u \in \Fr[y]$, $k > 0$, and species
  $\lambda(A) = \{(1; \lambda_{1}, \lambda_{2}, \dots,
  \lambda_{k})\}$. Then the number of $A$-invariant lines in an
  $n$-dimensional $\Fr$-vector space $V$ is
  \begin{equation}
  \label{eq:23a}
  \gf_{1} (A) = [s]_{r},
  \end{equation}
  where $s = \sum_{1 \leq j \leq k} \lambda_{j}$.
\end{lemma}

\begin{proof}
  For $v \in V \mysetminus \{0\}$, the following are equivalent
  for the line $\braket{v}$:
  \begin{itemize}
  \item $\braket{v}$ is $A$-invariant.
  \item $\braket{v}$ is in the eigenspace of the linear eigenfactor $u$
  (a factor of $A$'s minimal polynomial).
  \end{itemize}

  For a linear eigenfactor $u$, the eigenspace has dimension $\dim
  (\ker (u(A))) = \sum_{1 \leq j \leq k} \lambda_{j} = s$ and thus
  contains $(r^{s}-1)/(r-1)$ lines.
\end{proof}

With $\gf_{0} = 1$, \eqref{eq:24}, and \eqref{eq:27}, we now compute
$\gf_{1}$ for a rational Jordan form $A$ with arbitrary minimal polynomial.
\begin{proposition}
  \label{cor:1}
  Let $A \in \Fr^{\nxn}$ be in rational Jordan form as in
  \eqref{eq:jordanform} with species $\lambda(A) = \{ (\deg u_{i};
  \lambda_{i1}, \lambda_{i2}, \dots, \lambda_{i\ell_{i1}}) \colon 1 \leq i
  \leq t\}$. Then the number of $A$-invariant lines in 
  $\Fr^n$
  is
  \begin{equation}
  \label{eq:23}
  \gf_{1}(A) = \sum_{\substack{1 \leq i \leq t \\ \deg u_{i} = 1 }} [s_{i}]_{r},
  \end{equation}
  where $s_{i} = \sum_{1 \leq j \leq \ell_{i1}} \lambda_{ij}$ for $1
  \leq i \leq t$.
\end{proposition}

This answers question~\ref{it:q1} for $d = 1$. For $d > 1$, the number
$\gf_{d}$ of $d$-dimensional $A$-invariant subspaces can be derived
from the species with the formulas of \cite{fri11}. We make them
available through the \citetalias{sage611}-package accompanying this
paper.


  For perspective, formula \eqref{eq:23} allows us to determine exactly
the possible values for the number of right components of an additive
polynomial that have exponent $1$. By \autoref{cor:2}, this is
equivalent to finding the possible number of roots of certain
projective polynomials. Let
\begin{equation}
  \label{eq:29}
  M_{q,r,n,1} = \{ \# H_{1}(f) \colon f \in \Fq[x; r]_{n} \}
\end{equation}
be the set of possible numbers of right components of exponent $1$ for
monic squarefree $r$-additive polynomials of exponent $n$ over $\Fq$.

For a positive integer $m$, let $\Pi_{m}$ be the set of unordered partitions
(multisets)
$\pi = \{ \pi_{1}, \dots, \pi_{k} \}$ of $m$ with positive integers
$\pi_{i}$ and $\pi_{1} + \dots + \pi_{k} = m$. For any partition $\pi
\in \Pi_{m}$, we define the $r$-bracket $[\pi]_{r} = [\pi_{1}]_{r} +
[\pi_{2}]_{r} + \dots + [\pi_{k}]_{r}$.
Then \eqref{eq:23} yields the
following theorem.

\begin{theorem}
  \label{thm:S}
  Let $M_{n} = M_{q,r,n,1}$ be as in \eqref{eq:29} and define
  \begin{align}
    \Mmax_{0} & = \{0\}, \\
    \Mmax_{i} & = \Mmax_{i-1} \cup  \{ [\pi]_{r} \colon \pi \in \Pi_{m}\} 
  \end{align}
  for $1 \leq i \leq n$. Then $M_{n} \subseteq \Mmax_{n}$.
\end{theorem}

Generally, $M_{n} = \Mmax_{q,r,n,1}$ for all but a few triples
$(q, r, n)$, especially over small fields $\Fq$ where not all possible (similarity classes of)
Jordan forms may occur.
As an example, for $q=r = n =2$, we have merely two monic squarefree polynomials under
consideration. That is simply not enough to
cover all four cases in $\hat M_2$.
A list of the first seven values follows.
\begin{align*}
  \Mmax_{0} & = \{ 0 \}, \\
  \Mmax_{1} & =  \Mmax_{0} \cup \{ [1]_{r} \} = \{0, 1\}, \\
  \Mmax_2 & = \Mmax_{1} \cup \{ 2[1]_{r}, [2]_{r} \} = \{0, 1, 2, r
  + 1  \},~\mbox{(consistent with \cite{blu04a})}\\
  \Mmax_3 & =  \Mmax_2 \cup \{3, [2]_{r} + 1,  [3]_{r} \} \\
  & = \{ 0, 1, 2, 3, r+1, r+2, r^2+r+1 \}, \\
  \Mmax_4 & = \Mmax_3 \cup \{4,  [2]_{r}
  + 2, 2 [2]_{r}, [3]_{r} + 1, [4]_{r}\} \\
   & = \{ 0, 1, 2, 3, 4, r+1, r+2, r+3, 2r+2,  r^2+r+1,  r^2+r+2, \\
     & \quad\quad r^3+ r^2+r+1 \}, \\
  \Mmax_5 & = \Mmax_4 \cup \{5, [2]_{r} + 3,
  2 [2]_{r} + 1, [3]_{r} + 2,  [3]_{r} + [2]_{r}, 
  [4]_{r} + 1, [5]_{r} \} \\
  & = \{ 0, 1, 2, 3, 4, 5, r+1, r+2, r+3, r+4, 2r+2,  2r+3, \\
    & \quad\quad   r^2+r+1,  r^2+r+2,  r^2+r+3, r^2+2r+2, \\
      & \quad\quad r^3+ r^2+r+1,  r^3+ r^2+r+2,  r^4+r^3+ r^2+r+1 \}, \\
  \Mmax_6 & = \Mmax_5 \cup \{ 6, [2]_{r} + 4, 2 [2]_{r} + 2,
   3 [2]_{r},  [3]_{r} + 3, [3]_{r} + [2]_{r} + 1,  2 [3]_{r},\\
  & \quad \quad \quad  [4]_{r} + 2,
   [4]_{r} + [2]_{r},   [5]_{r} + 1, [6]_{r} \} \\
     & = \{ 0, 1, 2, 3, 4, 5, 6, r+1, r+2, r+3, r+4, r+5, 2r+2,  2r+3,  2r+4, 3r+3, \\
    & \quad\quad  r^2+r+1,  r^2+r+2,  r^2+r+3, r^2+r+4, r^2+2r+2,  r^2+2r+3, \\
    & \quad\quad 2r^2+2r+2,  r^3+ r^2+r+1,  r^3+ r^2+r+2,  r^3+ r^2+r+3, \\
      & \quad\quad r^3+ r^2+2r+2, r^4+r^3+ r^2+r+1,  r^4+r^3+ r^2+r+2,\\
        & \quad\quad r^5+ r^4+r^3+ r^2+r+1  \}. \\
\end{align*}
The size of $\Mmax_{n}$ equals $\sum_{0\leq k\leq n} p(k)$, where $p(k)$
is the number of additive partitions of $k$. For $n \to \infty$, $p(n)$
grows exponentially as $\exp (\pi \sqrt{2n/3})/(4n\sqrt{3})$
\citep{harram18}, but is still surprisingly small considering the
generality of the polynomials involved.

Concerning question~\ref{it:q2}, we recall that all maximal chains
\eqref{eq:25} have equal length by the Krull-Schmidt Theorem. Let
$A \in \Fr^{\nxn}$ be in rational Jordan form on $V$ and let
$\numchains{A}$ denote the number of all maximal $A$-invariant
chains \eqref{eq:25}. If the lattice is a grid, these are the binomial coefficients.

The number of $A$-invariant chains depends only on the species
$ \lambda(A)$ and we write
$\numchains{\lambda(A)} = \numchains{A}$. For
every minimal nonzero $A$-invariant subspace $U$, there is a
canonical bijection -- given by $/U$ and $\oplus U$ -- between the
chains for $V$ that start with $U_{1} = U$ and chains for $V/U$. Thus,
we have the recursion formula
\begin{equation}
  \label{eq:30}
  \numchains{\lambda(A)} = \sum_{
    \shortstack{\footnotesize minimal, nonzero \\
      \footnotesize $A$-invariant $U\subseteq V$} }
    \numchains{\lambda(A|_{V/U})},
\end{equation}
where $A|_{V/U}$ is $A$ taken as a linear transformation on
the quotient vector space $V/U$, of dimension $n - \dim(U)$.

We now have two tasks.
\begin{itemize}
\item Find all minimal nonzero $A$-invariant $U \subset V$.
\item Derive $\lambda(A|_{V/U})$ for each such $U$.
\end{itemize}

Every minimal nonzero $A$-invariant subspace $U \subseteq V$ is
contained in the eigenspace $V_{i} = \ker(u_{i}^{k_{i}}(A))$ for a
unique $ i \leq t$ and we can partition the formula
\eqref{eq:30} in the light of the vector space decomposition
\eqref{eq:9} as
\begin{equation}
  \label{eq:26}
  \numchains{\lambda(A)} = \sum_{\text{eigenfactors } u_{i}} \sum_{
    \shortstack{\footnotesize minimal, nonzero \\
      \footnotesize $A$-invariant $U\subseteq V_{i}$} }
    \numchains{\lambda(A|_{V/U})}.
\end{equation}

As for question~\ref{it:q1} above, we make two simplifications. First,
it is sufficient to study $A$ where
$\minpoly(A) = u_{1}^{k_{1}} \cdots u_{t}^{k_{t}}$ is the product of
linear $u_{i}$ by \eqref{eq:24} and \eqref{eq:27}. Second, we will
deal only with primary vector spaces, i.e. a single eigenfactor
$u_{i}$, and thus only the inner sum in \eqref{eq:26}.

\begin{example}
  \label{exa:4}
  We have the following base case. If $A = (J_{u}^{(\ell)})$ consists
  only of a single Jordan block, i.e.
  $\lambda = \{(1;0,\dots, 0, \lambda_{\ell} = 1)\}$, we have a unique
  maximal chain of $A$-invariant subspaces
\begin{equation}
  \label{eq:38}
  0 \subsetneq \braket{e_{1}} \subsetneq \braket{e_{1}, e_{2}}
  \subsetneq \dots \subsetneq V
\end{equation}
and $\numchains{A} = 1$. For completeness, we note that $U =
\braket{e_{1}}$ is the unique minimal nonzero $A$-invariant subspace,
$A|_{V/U} = (J_{u}^{(\ell-1)})$, and $\lambda(A|_{V/U}) =
\{(1;0,\dots, 0, \lambda_{\ell-1} = 1)\}$.
\end{example}

For $\lambda(A)= \{(1; \lambda_{1}, \dots, \lambda_{k})\}$, we already
know that the number of minimal nonzero $A$-invariant subspaces is
$[\sum_{1 \leq i \leq k} \lambda_{i}]_{r}$ from \eqref{eq:27} and
\eqref{eq:23a}. We need to scrutinize them further.
For
\begin{equation}
  \label{eq:39}
  A = \diag(J_{u}^{(\ell_{1})}, J_{u}^{(\ell_{2})}, \dots, J_{u}^{(\ell_{s})})
\end{equation}
with $u = y-a$,
$\ell_1 \geq \cdots \geq \ell_s$, $\minpoly(A) = u^{\ell_1}$,  $s = \sum \lambda_{j}$,
and $\lambda_{j'} = \# \{\ell_{j} = j' \colon 1 \leq j
\leq s \}$, we re-index the basis of $V$ as
\begin{equation}
  \label{eq:40}
  e_{11}, \dots, e_{1\ell_{1}}, e_{21}, \dots, e_{2\ell_{2}}, \dots,
  e_{s1}, \dots, e_{s\ell_{s}}.
\end{equation}
The $d$-dimensional eigenspace is
$\ker u(A) = \braket{e_{11}, e_{21}, \dots, e_{s1}}$ and contains
$[s]_{r}$ lines, that is, $1$-dimensional subspaces, and these are the only
minimal non-zero subspaces.

Let $U$ be an $A$-invariant subspace. We define its \emph{support}
$\supp(U)$ (in the basis \eqref{eq:40}) as the set of all base vectors
for which $e_{ij} \cdot U \neq 0$. For a minimal, that is, $1$-dimensional,
$U$, we have $j=1$ for all $e_{ij}$ in its support, since these are
the base vectors that span the eigenspace.

The support links the subspace $U$ to the Jordan blocks that act
non-trivially on $U$. Of particular interest are the Jordan blocks of
minimal size that act non-trivially on $U$. We define
\begin{equation}
  \label{eq:42}
  \depth{U} = \min \{\ell_{j} \colon e_{j1} \in \supp(U)\}.
\end{equation}
Note that there may be several Jordan blocks of size $\depth{U}$
acting on the support of $U$.

\begin{example}
  For $A =       \begin{pmatrix}
        a & 1 & \\
        & a & \\
        & & a
      \end{pmatrix} $, we have $\braket{e_{1}}$ of depth 2 and
      $\braket{e_{1} + \alpha e_{3}}$ for $\alpha \in \Fr$ of depth
      1. And these are all $r+1$ nonzero minimal $A$-invariant subspaces.
\end{example}

To make \eqref{eq:30} applicable, we now determine the number of
minimal nonzero $A$-invariant subspaces of depth $j$ for
$1 \leq j \leq k$. Let $\lambda = (1; \lambda_{1}, \dots,\lambda_{k})$
be the species of the eigenvalue under consideration.  The possible
values for the depth of a nonzero minimal $A$-invariant subspace
range from $1$ to $k$, where $k = \max{\ell_{j}}$ and the following
counting formula follows easily by inclusion-exclusion.

\begin{proposition}
  \label{pro:3}
  Let $A$ be primary on $V$, with species $\lambda(A) = \{(1;
  \lambda_{1}, \dots, \lambda_{k})\}$.
  \begin{ronumerate}
\item The number of $A$-invariant subspaces with depth $i$ is
\begin{equation}
  \label{eq:43}
  \numdepth{\lambda, i} = r^{\lambda_{i+1} + \dots +
    \lambda_{k}} [\lambda_{i}]_{r}.
\end{equation}
\item Let $U$ be an $A$-invariant subspace with depth $i$. Then $A$ is well-defined on $V/U$ and has species
  \begin{equation}
    \label{eq:31}
    \lambda(A|_{V/U}) = \lambda_{\hat{i}} = \begin{cases*}
      (1; \lambda_{1}-1, \lambda_{2}, \dots, \lambda_{k}) & if $i =
      1$, \\
(1; \lambda_{1}, \dots, \lambda_{i-1}+1,
    \lambda_{i}-1, \dots, \lambda_{k}) & otherwise.
    \end{cases*}
  \end{equation}
\item\label{it:33} The number of maximal $A$-invariant chains is given by the
 recursion
\begin{align}
  \numchains{\{(1;1)\}} & = 1, \\
  \numchains{\lambda(A)} & = \sum_{1 \leq j \leq k} \numdepth{\lambda, j} \cdot
  \numchains{ \lambda_{\hat{j}} }.
\end{align}
\end{ronumerate}
\end{proposition}

\begin{proof}
  \begin{ronumerate}
  \item For $i = k$, we have $[\lambda_{k}]_{r}$ eigenspaces of
    depth $k$. For $i < k$, we have $[\lambda_{i} + \lambda_{i+1} + \dots +
    \lambda_{k}]_{r}$ eigenspaces of depth at least $i$ and we find by
    direct computation
    \begin{align}
      \label{eq:41}
      \numdepth{\lambda, i} & = [\lambda_{i} + \lambda_{i+1} + \dots +
    \lambda_{k}]_{r} - [\lambda_{i+1} + \dots +
                              \lambda_{k}]_{r} \\
      & = \frac{r^{\lambda_{i} + \lambda_{i+1} + \dots + \lambda_{k}}
        - 1}{r-1} - \frac{r^{\lambda_{i+1} + \dots + \lambda_{k}} -
        1}{r-1} \\
      & = r^{\lambda_{i+1} + \dots +
    \lambda_{k}} \cdot [\lambda_{i}]_{r},
    \end{align}
    as claimed.

\item We dealt with the base case in \autoref{exa:4}.

  Without loss of generality, we assume that $e_{11}$ is in the
  support of $U$ and that the corresponding first Jordan block has
  size equal to the depth of $U$. We have
  \begin{equation}
    \label{eq:44}
    U = \braket{e_{11} + \alpha_{2}e_{21} + \dots + \alpha_{s}e_{s1}}
  \end{equation}
  for some $\alpha_{j} \in \Fr$ and $\alpha_{j} \neq 0$ only if the
  corresponding Jordan block is larger than the first one. We turn
  \eqref{eq:40} into the following basis for $V/U$:
\begin{gather}
  e_{12} + \alpha_{2} e_{22} + \dots + \alpha_{s} e_{s2} + U, \\
  e_{13} + \alpha_{2} e_{23} + \dots + \alpha_{s} e_{s3} + U, \\
  \vdots \\
  e_{1\ell_{1}} + \alpha_{2} e_{2\ell_{1}} + \dots + \alpha_{s} e_{s\ell_{1}} + U, \\
  e_{21} + U, \dots, e_{2\ell_{2}} + U, \\
  \vdots \\
  e_{s1} + U, \dots, e_{s\ell_{s}} + U.
\end{gather}
  In other words, we drop the projection of the first base vector (due
  to the linear dependence introduced by $U$) and modify the base
  vectors for the first Jordan block. A direct computation shows that
  $A|_{V/U}$ is in rational Jordan form, its first Jordan block is
  equal to the first Jordan block of $A$ reduced by size 1, and all
  other Jordan blocks ``remained'' unchanged.
  \item This follows from \eqref{eq:26} using \ref{it:1} and \ref{it:2}.
  \end{ronumerate}
\end{proof}

In the general case of several eigenfactors we obtain $\numchains{A}$
by \eqref{eq:26} using the formulae in \autoref{pro:3}~\ref{it:33} for
each eigenfactor.

\section{The Frobenius automorphism on the root space}
\label{sec:Frob}

\todo{INFO $q$ a power of $r$}

In this section, we present an efficient algorithm to compute the
rational Jordan form of the Frobenius automorphism on the root space of a
squarefree monic additive polynomial $f$. With the results of
\autoref{sec:count-invar-subsp}, this yields the number of right
components of $f$ of a given degree. The straightforward approach
suffers from possibly exponential costs for the description of the
root space $V_{f}$, see \autoref{exa:2}.

The \emph{centre} of the Ore ring $\Fq [x; r]$ will be a useful
tool. For $q$ a power of $r$, so that $\Fr \subseteq \Fq$, it equals
\begin{equation}
  \label{eq:2}
   \Fr [x; q] = \bigl\{\sum_{0\leq i\leq n} a_ix^{q^{i}} \colon n \in
  \ZZ_{\geq 0}, \quad a_0,\ldots,a_{n} \in \Fr \bigr\} \subseteq \Fq [x; r],
\end{equation}
see, e.g., \citet[Section~3]{gie98}. Every element $f\in\Fq[x;r]$ has
a unique \emph{minimal central left component} $\fmclc\in\cAq$, the
unique monic polynomial in $\cAq$ of minimal degree such that $\fmclc=g \circ
f$ for some nonzero $g\in\Fq[x;r]$. For squarefree $f$, it is the
monic generator of the largest two-sided ideal $I(f)$ contained in the
left ideal generated by $f$. The ideal $I(f)$ is then known as the
\emph{bound} of $f$, see
\citet[page~83]{jac43}. 

\begin{fact}[{\citealt[Lemma~4.2]{gie98}}]
  \label{fact:mclm}
  Let $r$ be a power of a prime $p$ and $q=r^d$.  For $f\in\Fq[x;r]$
  of exponent $n$, we can find its minimal central left component
  $\fmclc \in \cAq$ with $O(n^{3}d \MM(d) + n^{2}d^{2}\MM(d)\log d)
  \subseteq \softOh{n^{3}d^{2} + n^{2}d^{3}}$ operations in $\Fr$, where
  $\MM(d)$ is the number of operations to multiply two polynomials
  over $\Fr$ with degree at most $d$ each.
\end{fact}

The ``schoolbook'' method gives $\MM(d) = O(d^{2})$ and \cite{harhoe19}
show $\MM(d) = O(d \log d \, 4^{\log^{*} d})$.  The recent, as yet
unpublished, preprint of \cite{harhoe19b} claims $\MM(d) = O(d \log d)$, which many
consider to be the best achievable asymptotic bound.

\citet[Theorem~II.3.2]{bor12} gives an algorithm for $\fmclc$
with $\softOh{n^{\omega}d^{\omega} + n^{2}d^{2}\log r}$ operations in
$\Fr$, where $d$ and $n$ are as above and $\omega$ is an exponent of
square matrix multiplication over $\Fr$.

The centre $\Fr[x; q]$ is a commutative subring of $\Fq[x; r]$ and
isomorphic to $\Fr[y]$ with the usual addition and multiplication via
\begin{equation}
\tau \colon
  \begin{aligned}
\Fr[x;q] & \to \Fr[y], \\
f = \sum_{0\leq i\leq n} a_ix^{q^{i}} & \mapsto \tau(f) =
\sum_{0\leq i\leq n} a_iy^i,
  \end{aligned}
\end{equation}
see \citet[pages~24--25]{mcd74}. The isomorphic image $\Fr[y]$ is a
unique factorization domain and factorizations in $\Fr[y]$ are in
one-to-one correspondence with decompositions in $\Fr[x; q]$ into
central components. The following main theorem shows the close
relationship between the minimal central left component of an
additive polynomial and the minimal polynomial of the Frobenius
automorphism on its root space.

\begin{theorem}
  \label{thm:minpolymclm} Let $r$ be a power of a prime $p$ and $q$ a
  power of $r$. Let $f \in \pP[n]{\Fq}{r}$ be monic squarefree of
  exponent $n$ with root space $V_{f} \subseteq \Fqbar$ and minimal
  central left component $\fmclc \in \Fr[x; q]$. Then the image
  $\tau(\fmclc)\in\Fr[y]$ is the minimal polynomial of the $q$th power
  Frobenius automorphism $\sigma_{q}$ on the $\Fr$-vector space~$V_{f}$.
\end{theorem}

\begin{proof}
  For a central $g = \sum_{0 \leq i \leq m} g_{i} x^{q^{i}} \in \Fr[x;
  q]$, we have $\tau(g) = \sum_{0 \leq i \leq m} g_{i} y^{i}$ $\in
  \Fr[y]$ and $(\tau(g)) (\sigma_q) = g$,
   and the following are equivalent:
  \begin{itemize}
  \item $g$ is a right or left component of $f$;
  \item $g(\alpha) = 0$ for all $\alpha \in V_{f}$;
  \item $(\tau(g)(\sigma_{q}))(\alpha) = 0$ for all $\alpha \in V_{f}$.
  \end{itemize}
  The first two items are equivalent by \eqref{eq:10} and the
  squarefreeness of $f$ and since $g$ is central. The last two items are equivalent 
  since $\tau(g) (\sigma_q) = g$.

  Thus, $g$ is a central left component
  of $f$ if and only if $\tau(g)$ annihilates $\sigma_{q}$ on
  $V_{f}$. Since $f^{*}$ and the minimal polynomial of $\sigma_{q}$
  are the unique monic polynomials of minimal degree with these
  properties, respectively, we have the claimed equality.
\end{proof}

It is useful to recall a little more about the ring $\Fq[x;r]$.
\cite{ore33b} shows that for any $f,g\in\Fq[x;r]$, there exists a
unique monic $h\in\Fq[x;r]$ of maximal degree, called the
\emph{greatest common right component} ($\gcrc$) of $f$ and $g$, such
that $f = u \circ h$ and $g = v \circ h$ for some $u, v \in
\Fq[x;r]$. Also, $h = \gcrc(f,g) = \gcd(f,g)$, and the roots of $h$
are those in the intersection of the roots of $f$ and $g$, in other
words $V_{\gcrc (f,g)} = V_{f} \cap
V_{g}$. 
In fact, there is an efficient Euclidean-like algorithm for computing
the $\gcrc$; see \cite{ore33b} and \cite{gie98} for an
analysis.
The usual Euclidean algorithm for $\gcd(f,g)$ is insufficient,
since the degrees of $f$ and $g$ may be exponential in their exponents.

\begin{fact}[{\citealt[Lemma~2.1]{gie98}}]
  \label{fact:gcrclclc}
  Let $r$ be a power of a prime $p$ and $q = r^{d}$. For $f, g \in
  \Fq[x;r]$ of exponent $n$, we can find $\gcrc(f,g) \in
  \Fq[x;r]$ with $O(n^{2} \MM(d) d \log d) \subseteq \softOh{n^{2}
  d^{2}}$ operations in $\Fr$, where $\MM(d)$ is as in \autoref{fact:mclm}.
\end{fact}

\subsection{A fast algorithm for the rational Jordan form of
  $\sigma_{q}$ on $V_{f}$}

We now determine the rational Jordan form of the Frobenius
automorphism on the root space of an additive polynomial. We begin
with a factorization of the minimal polynomial and then compute every
eigenfactor's species independently. The following proposition deals with the base case,
where the minimal polynomial is the power of an irreducible
polynomial.

\begin{proposition}
  \label{pro:eigenspaces}
  Let $r$ be a power of a prime $p$, $q$ a power of $r$, $f \in
  \Fq[x;r]_{n}$ monic squarefree of exponent $n$ with minimal central
  left component $f^{*} \in \Fr[x;q]$, and $\sigma_{q}$ the $q$th
  power Frobenius automorphism on $V_{f}$. If $\tau(f^{*}) = u^{k}$
  for an irreducible $u \in \Fr[y]$ and $k > 0$, then
  \begin{align}
  \tau^{-1}(u^{j})) & = u^j(\sigma_q), \label{eq:16a}\\
    \ker(u^{j}(\sigma_{q})) & = V_{\gcrc(f, \tau^{-1}(u^{j}))}, \label{eq:20} \\
    \prod_{\alpha \in \ker(u^{j}(\sigma_{q}))} (x - \alpha) & =
    \gcrc(f, \tau^{-1}(u^{j})) \label{eq:16}
  \end{align}
  for all $j$ with $0 \leq j \leq k+1$, where $u^{0}(\sigma_{q})$ is the
  identity on $V_{f}$.
\end{proposition}

\begin{proof}
  Let $0 \leq j \leq k+1$. If we write $u^j = \sum_i w_i y^i$
  with all $w_i \in \Fr$, then
  $\tau^{-1} (u^j) = \sum_i w_i x^{q^i} = u^j(\sigma_q)$, which is \eqref{eq:16a}.
  The kernels of these two maps on $V_f$ form the same subset of $V_f$, so that
  $V_{\tau^{-1}(u^{j})} \cap V_{f} =
  V_{\gcrc(f, \tau^{-1}(u^{j}))}$. This shows \eqref{eq:20}.

  Furthermore, the bijection $\varphi_{\dim ( \ker (u^{j} (\sigma_{q})))}$
  from \eqref{eq:7} maps the left and right hand sides of
  \eqref{eq:20} to the
  left and right hand sides of \eqref{eq:16}, respectively.
\end{proof}

\begin{corollary}
  \label{cor:3}
  In the notation and under the assumption of
  \autoref{pro:eigenspaces}, let $u$ be irreducible of degree $m$ and
  $\nu_{j} = \expn(\gcrc(f, \tau^{-1}(u^{j})))$ for $0 \leq j \leq k +
  1$. Then the species of the rational Jordan form of $\sigma_{q}$ on
  $V_{f}$ is $\{(m; \lambda_{1}, \lambda_{2}, \dots, \lambda_{k})\}$,
  where
  \begin{equation}
    \label{eq:18}
    \lambda_{j} = (2 \nu_{j} - \nu_{j-1} - \nu_{j+1})/m,
  \end{equation}
  for $1 \leq j \leq k$.
\end{corollary}

\begin{proof}
  For monic squarefree $g \in \Fq[x;r]$, we have $\expn g = \dim
  V_{g}$ due to the bijection~\eqref{eq:6}. For $0 \leq j \leq k+1$,
$\gcrc(f, \tau^{-1}(u^{j}))$ is monic squarefree and thus
  \begin{equation}
    \label{eq:21}
    \nu_{j} = \dim(V_{\gcrc(f, \tau^{-1}(u^{j}))}) = \dim (\ker (u^{j}
  (\sigma_{q}))) = \nul(u^{j}(S))
  \end{equation}
  by \eqref{eq:20}. The claim follows from \eqref{eq:19}.
\end{proof}

In the case of a minimal polynomial with arbitrary factorization, we
treat every eigenfactor separately with \autoref{cor:3}, see
\citet[Theorem~4.1]{gie98}. The result is \autoref{algo:ratJNF}. It
computes the rational Jordan form of the Frobenius automorphism on the
root space of a given $f \in \Fq[x; r]_{n}$.

\begin{algorithm2e}[h]
  \caption{\texttt{RationalJordanForm}}
  \label{algo:ratJNF}

  \KwIn{$r$-additive monic squarefree $f\in\Fq[x;r]_{n}$ of exponent
    $n$, where $q = r^{d}$ and $r$ is a power of a prime $p$}
  \KwOut{rational Jordan form $S \in \Fr^{\nxn}$ as in
    \eqref{eq:jordanform} of the $q$th power Frobenius automorphism on
    $V_{f}$}

  $f^{*} \gets$ minimal central left component of $f$ \label{step:1}\;
  $u_{1}^{k_{1}} u_{2}^{k_{2}} \cdots u_{t}^{k_{t}} \gets$
  factorization of $\tau(\fmclc)$ into pairwise distinct monic
  irreducible $u_{i} \in \Fr[y]$ with $k_{i} > 0$ for $1 \leq i \leq t$\label{step:2}\;
  $S \gets \varnothing$ \label{step:3} \tcp*{initialize ``empty matrix''}
  \For{$i \gets 1$ \To $t$ \label{loop:outer}}{
    \tcp{determine the species of $u_{i}$}
    \For{$j \gets 0$ \To $k_{i} + 1$ \label{loop:inner1}}{
      $h_{j} \gets \gcrc(f, \tau^{-1}(u_{i}^{j}))$ \label{step:7}\;
      $\nu_{j} \gets \expn h_{j}$ \label{step:8}
      \tcp*{equal to $\nul(u_{i}^{j}(S))$}
    }
    $m \gets \deg_{y} u_{i}$ \label{step:5}\;
    \For{$j \gets 1$ \To $k_{i}$ \label{loop:inner2}}{
      $\lambda_{j} \gets (2 \nu_{j} - \nu_{j-1} - \nu_{j+1})/m$ \label{step:11}
      \tcp*{employ \eqref{eq:18}}
      $S \gets \diag(S,
      \underbrace{J_{u_{i}}^{(j)},\ldots,J_{u_{i}}^{(j)}}_{\lambda_{j}
        \text{-times}})$ \label{step:12} \tcp*{append Jordan blocks}
\label{step:13}
    }
\label{step:14}
  }
  \Return $S$ \label{step:return}
 \end{algorithm2e}

\begin{theorem}
  \label{thm:ratJNF-correct}
  \autoref{algo:ratJNF} works correctly as specified and takes an
  expected number of $\softOh{n^{3}d^{4}}$ field operations in $\Fr$.
\end{theorem}

\begin{proof}
  The notation in the algorithm corresponds to that of the rational
  Jordan form \eqref{eq:jordanform} and \autoref{cor:3}.  In Step~1,
  we know from Theorem \ref{thm:minpolymclm} that $\fmclc$ is the
  minimal polynomial of $S$.  Therefore all rational Jordan blocks
  correspond to factors of $\fmclc$ (determined in Step~2) and we only
  need to figure out every eigenfactor's species. By
  \citet[Theorem~4.1]{gie98}, we can treat every eigenfactor
  separately (Steps~\ref{loop:outer}--\ref{step:14}) and align the
  resulting rational Jordan blocks along the main diagonal
  (Step~\ref{step:12}, initialized in Step~\ref{step:3}).

  For every eigenfactor $u_{i}$ the first inner loop
  (Steps~\ref{loop:inner1}--\ref{step:8}) determines $\nu_{j}$ as defined in
  \autoref{cor:3} for $0 \leq j \leq k_{i}+1$. The second inner loop
  (Steps~\ref{loop:inner2}--\ref{step:12}) derives the number $\lambda_{j}$ of
  rational Jordan blocks of order $j$ for $u_{i}$ (Step~\ref{step:11})
  via formula \eqref{eq:18} and extends $S$ along its main diagonal
  accordingly (Step~\ref{step:12}).

  Doing this for all eigenfactors and all possible orders returns the
  specified output in Step~\ref{step:return}.

  We assume that the isomorphism $\tau$ and its inverse are free
  operations. If the polynomials are stored as vectors of
  coefficients, these operations merely change the way this
  information is interpreted. We also take for granted a free
  operation to determine the exponent of an additive and the degree of
  an ``ordinary'' polynomial in Steps~\ref{step:8} and \ref{step:5},
  respectively. Finally, we neglect the (cheap) integer arithmetic in
  Step~\ref{step:11}.

  Step~\ref{step:1} uses $\softOh{n^{3}d^{2} + n^{2}d^{3}}$ field
  operations in $\Fr$, see \autoref{fact:mclm}. We have $\expn f^{*}
  \leq dn$ and thus $\deg_{y} \tau(f^{*}) \leq n$. The factorization
  in Step~\ref{step:2} can be done in random polynomial time with
  $\softOh{n^{2} + n \log r}$ field operations in $\Fr$, see,
  e.g. \citet[Corollary~14.30]{gatger13}. The worst case occurs when
  $\tau(f^{*})$ 
  is the $n$th power of a linear eigenfactor $u$. The $n + 2$ powers
  of $u$ can be obtained with $\softOh{n^{2}}$ field operations in
  $\Fr$. The additive polynomial $\tau^{-1}(u^{j})$ has exponent $dj$
  and each $\gcrc$ in Step~\ref{step:7} requires $\softOh{\max(n,
    dj)^{2}d^{2}} \subseteq \softOh{n^{2}d^{4}}$ field operations in
  $\Fr$, see Step~\ref{fact:gcrclclc}.  The complete inner loop thus
  requires $\softOh{n^{3}d^{4}}$ field operations which dominates the
  costs of the previous steps.
\end{proof}

Only the distinct-degree factorization in Step~\ref{step:2} requires
randomization. But this granularity is necessary for our approach as
the following example shows. Let
\begin{equation}
  \label{eq:5}
  A = \begin{psmallmatrix}
a &&& \\
& a && \\
&& a & \\
&&& b
\end{psmallmatrix},
\hspace*{10pt}
B = \begin{psmallmatrix}
a &&& \\
& a && \\
&& b & \\
&&& b
\end{psmallmatrix} \in \Fr^{4 \times 4},
\end{equation}
with distinct nonzero $a, b \in \Fr$. Then $A$ and $B$ are two
rational Jordan forms with distinct species $\{(1;3),(1;1)\}$ and $\{2
\times (1;2)\}$, respectively, but equal minimal polynomial $u =
(y-a)(y-b)$. The single equal-degree factor has multiplicity $1$ and
yields only the information $\dim \ker u(A) = \dim \ker u(B) = 4$,
that is the sum of orders of blocks corresponding to eigenfactors of
degree~$1$.

\cite{carleb17} give an algorithm for the species of the Frobenius
operator on the $n$-dimensional module $\Fq[x;r] / (\Fq[x;r] \cdot f)$,
as in \cite{gatgie09},
and count complete decompositions, as in \cite{fri11}.
Related counting problems are also considered in \cite{bor12}.

The costs of \autoref{algo:ratJNF} are only polynomial in $\expn f$
and $\log q$, despite the fact that the actual roots of $f$ may lie in
an extension of exponential degree over $\Fq$ as illustrated in the
following example and \autoref{fig:roads}.

\begin{example}
  \label{exa:2}
  Let $q = r$ and $f \in \Fq[y]$ be primitive of degree $n$. Its
  \emph{additive $q$-associate}
  $\tau^{-1}(f)$ factors into $x$ and the
  irreducible $\tau^{-1}(f)/x$ of degree $q^{n}-1$ over $\Fq$, see
  \citet[Theorem~3.63]{lidnie97}. Thus, the splitting field of the
  additive $\tau^{-1}(f)$ is an extension of
  $\Fq$ of degree $q^{n} - 1$.
\end{example}

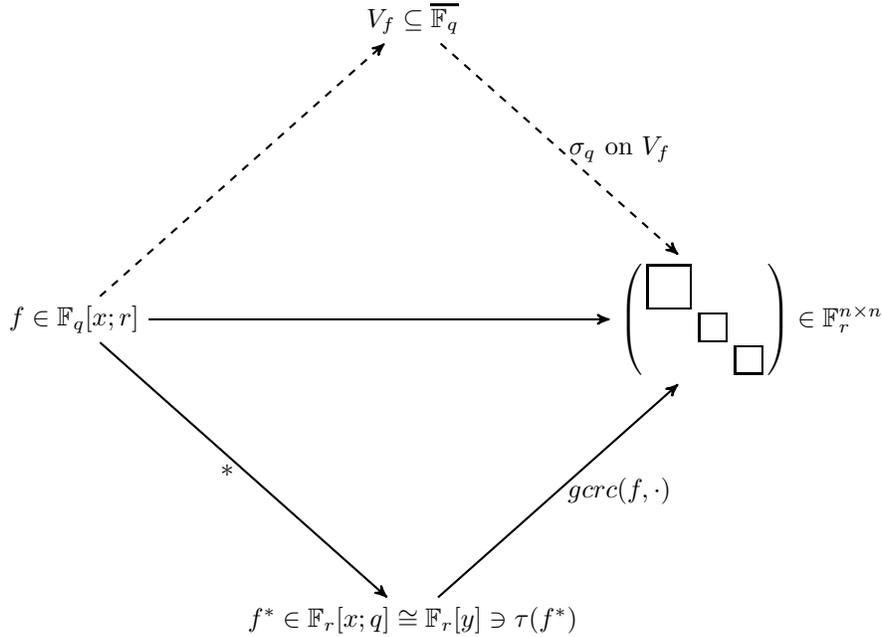
\begin{figure}[h]
  \centering
    \begin{tikzpicture}[
map/.style={thick, ->, >=stealth'}
]
\begin{scope}[
xscale =4.5,
yscale =2
]
    \tikzstyle{every node} = [rectangle]
    \node (f) at (-1,0) {$f \in \Fq[x;r]$};
    \node (Vf) at (0,2) {$V_{f} \subseteq \Fqbar$};
    \node (fstar) at (0,-2) {$f^{*} \in \Fr[x;q] \cong \Fr[y] \ni \tau(f^{*})$};
    \node (S) at (1,0) {$\left(
\begin{smallmatrix}
\framebox{\hbox to 10pt{\vbox to 10pt{}}}  & & \\
& \framebox{\hbox to 4pt{\vbox to 4pt{}}} & \\
& & \framebox{\hbox to 4pt{\vbox to 4pt{}}}
\end{smallmatrix}
\right)
\in \Fr^{\nxn}$};
    \draw[map,style=dashed] (f) -- (Vf) node[pos=.5,left] {
};
    \draw[map,style=dashed] (Vf) -- (S) node[pos=.5,right] {$\sigma_{q}$ on $V_{f}$
};
    \draw[map] (f) -- (S) node[pos=.5] {
};
     \draw[map
] (f) -- (fstar) node[pos=.5,left] {$*$
};
     \draw[map
] (fstar) -- (S) node[pos=.5,right] {$\gcrc(f, \cdot)$
};
\end{scope}
\end{tikzpicture}
\caption{\autoref{algo:ratJNF} computes the rational Jordan form of
  the Frobenius automorphism on the root space $V_{f}$ of $f$ while
  avoiding the expensive computation (dashed) of and on the root space
  itself.}
  \label{fig:roads}
\end{figure}

Together with the results of \autoref{sec:count-invar-subsp}, we can
now count the number of irreducible right components of degree $r$ of
any $r$-additive polynomial $f\in\Fq[x;r]$ of exponent $n$. This also
yields a fast algorithm to compute the number of certain factors and
right components of projective and subadditive polynomials as
described in \autoref{sec:proj-subadd}.

\begin{example}
  \label{exa:3}
  \cite{bouulm14} build self-dual codes from factorizations of
  $x^{r^{n}}-ax$ beating previously known minimal distances. Over
  $\FF_{4}[x; 2]$, they exhibit $3$, $15$, $90$, and $543$ complete
  decompositions for $x^{2^{2}}+x$, $x^{2^{4}}+x$, $x^{2^{6}}+x$, and
  $x^{2^{8}}+x$, respectively.
\end{example}

\todo{@KZ: table for more and table for number of right components for this
  family}

\todo{@KZ: We should be able to produce these counts ``at lightning speed''.}

In this section, we assume the field size $q$ to be a power of
the parameter $r$.  As in \citeauthor{blu04a}'s (\citeyear{blu04a})
work, our methods go through for the general situation, where
$q=p^{d}$ and $r=p^{e}$ are independent powers of the characteristic.
Then $\Fq \cap \Fr = \Fs$ for $s =
p^{\gcd(d,e)}$ and the centre of $\Fq[x; r]$ is $\Fs[x; t]$ for $t =
p^{\lcm(d,e)}$.

\section{Conclusion and open questions}
\label{sec:conclusion}

We investigated the structure and number of all right components of an
additive polynomial. This involved three steps:
\begin{itemize}
\item a bijective correspondence between decompositions of an additive
  polynomial $f$ and Frobenius-invariant subspaces of its root space
  $V_{f}$ in an algebraic closure of $F$ (\autoref{sec:InNo}),
\item a description of the $A$-invariant subspaces of an $F$-vector space
  for a rational Jordan form $A \in F^{\nxn}$
  (\autoref{sec:rati-norm-forms}), and
\item an efficient algorithm for the rational Jordan form of the
  Frobenius automorphism on $V_{f}$ (\autoref{sec:Frob}). Its runtime
  is polynomial in $\log_{p}(\deg f)$.
\end{itemize}
A combinatorial result of \cite{fri11} counts the relevant
Frobenius-invariant subspaces of $V_{f}$ and thus our decompositions
(\autoref{sec:count-invar-subsp}). We also count the number of maximal
chains of Frobenius-invariant subspaces and thus the complete
decompositions.

In \autoref{thm:S}, we describe the small set of possible values
for the number of right components of exponent $r$ of a given additive
polynomial. The natural ``inverse'' question asks for the number of
additive polynomials that admit a given number of right
components. \todo{@vzG: cite forward Bluher?}

The root space $V_f$ has exponentially (in the exponent of $f$) many elements,
and the field over which it is defined may have exponential degree.
The efficiency of our algorithms in \autoref{sec:Frob} is mainly achieved
by avoiding any direct computation with $V_f$.

\remove{
\begin{figure}
\centering
\begin{tikzpicture}[scale = 2]
    \tikzstyle{every node} = [rectangle]
    \node (K) at (0,2) {$\Fqbar$}
     node (Fq) at (0,0.5) {$\Fq = \FF_{r^{d}}$}
     node (Fr) at (0,-0.5) {$\Fr$}
     node (Fp) at (0,-2) {$\Fp$};
    \draw (K) -- (Fq);
    \draw (Fq) -- (Fr);
    \draw (Fr) -- (Fp);
\end{tikzpicture} \quad
\begin{tikzpicture}[scale = 2]
    \tikzstyle{every node} = [rectangle]
    \node (K) at (0,2) {$\Fqbar$}
     node (Ft) at (0,1) {$\Ft = \FF_{p^{\lcm(d,e)}}$}
     node (Fq) at (-.5,0) {$\Fq = \FF_{p^{d}}$}
     node (Fr) at (.5,0) {$\Fr = \FF_{p^{e}}$}
     node (Fs) at (0,-1) {$\Fs = \FF_{p^{\gcd(d,e)}}$}
     node (Fp) at (0,-2) {$\Fp$};
    \draw (K) -- (Ft) -- (Fq) -- (Fs) -- (Fp);
    \draw (K) -- (Ft) -- (Fr) -- (Fs) -- (Fp);
\end{tikzpicture}
\caption{Lattice of subfields for $q$ a
  power of $r$ (left) and in the general case (right).}
\label{fig:fields}
\end{figure}
}
\section{Acknowledgments}

This work was supported by the German Academic Exchange Service (DAAD) in
the context of the German-Canadian PPP program. In addition, Joachim von
zur Gathen and Konstantin Ziegler were supported by the B-IT Foundation and
the Land Nordrhein-Westfalen. Mark Giesbrecht acknowledges the support of
the Natural Sciences and Engineering Research Council of Canada (NSERC).
Cette recherche a été financée par le Conseil de recherches en sciences
naturelles et en génie du Canada (CRSNG).
\vspace{0.5cm}


\bibliographystyle{elsarticle-harv}

\input{inv_subspaces-bib}

\end{document}

%% file: inv_subspaces-bib.tex
 \providecommand{\ymd}[3]{\csname @ifempty\endcsname{#1}{?#1/#2/#3?}{\csname
  @ifempty\endcsname{#2}{?#1/#2/#3?}{\csname
  @ifempty\endcsname{#3}{?#1/#2/#3?}{{\relax \day=#3\relax \month=#2\relax
  \year=#1\relax \number\day~\ifcase\month\or January\or February\or March\or
  April\or May\or June\or July\or August\or September\or October\or November\or
  December\fi \space\ifnum\year>0\relax \number\year \else \csname
  count@\endcsname1\relax \expandafter\advance\csname count@\endcsname-\year
  \expandafter\number\csname count@\endcsname~BC\fi}}}}}
  \providecommand{\hide}[1]{.} \providecommand{\Hide}[1]{\unskip}
  \makeatletter\def\eattilde{\@ifnextchar~{\@gobble}{}}\makeatother
  \providecommand{\HIde}[1]{\protect\HIdex{#1}}
  \providecommand{\HIdex}[1]{\protect\HIdey{#1}}
  \providecommand{\HIdey}[1]{\Hide{#1}\aftergroup\eattilde}
  \providecommand{\HIDE}[1]{\unskip\aftergroup\eattilde}
  \providecommand{\gobble}[1]{} \providecommand{\todo}[1]{\textbf{`?`?`?#1???}}
  \providecommand{\Name}[1]{#1} \providecommand{\Textgreek}[1]{\textgreek{#1}}
  \providecommand{\at}{\char64\relax} \providecommand{\altcite}[2][]{#1}
  \providecommand{\cyr}{\PackageError{cyrillic}{Package not loaded. Use
  \string\usepackage{cyrillic} to define \string\cyr\space appropriately}{}
  \gdef\cyr{\def\cprime{c'}{\bf ?cyr?}}\cyr}
  \makeatletter\protected@write\@auxout{}{\string
  \gdef\string\abbr{\string\csname\space
  @gobble\string\endcsname}}\gdef\abbr{}\makeatother
  \makeatletter\protected@write\@auxout{}{\string
  \gdef\string\bibliographyonly{\string\protect\string\abbr}}\gdef\bibliographyonly{\protect\abbr}\makeatother
  \makeatletter\protected@write\@auxout{}{\string
  \gdef\string\bodyonly{}}\gdef\bodyonly#1{}\makeatother